\documentclass[a4paper,UKenglish]{article}

\usepackage{amsmath}
\usepackage{amsfonts}
\usepackage{graphicx,tikz}
\RequirePackage{fancyhdr}
\usepackage{xcolor}
\usepackage{boxedminipage}

\usepackage{todonotes}

\usepackage{vmargin}
\setmarginsrb{1.1in}{1.1in}{1.1in}{1.1in}{0mm}{0mm}{0mm}{7mm}
\usepackage{amssymb,amsmath,amsthm}

\usepackage[linesnumbered, boxed, algosection]{algorithm2e}
\newcommand{\inv}{^{-1}}

\usepackage{amstext}

\newcommand{\defparproblem}[4]{
  \vspace{1mm}
\noindent\fbox{
  \begin{minipage}{0.96\textwidth}
  \begin{tabular*}{\textwidth}{@{\extracolsep{\fill}}lr} #1  & {\bf{Parameter:}} #3 \\ \end{tabular*}
  {\bf{Input:}} #2  \\
  {\bf{Question:}} #4
  \end{minipage}
  }
  \vspace{1mm}
}

\usepackage{graphics}
 \newtheorem{definition}{Definition}
 \newtheorem{claim}{Claim}

\usepackage{pgf}
\usepackage{paralist} 
\usepackage{graphicx}
\usepackage{boxedminipage}
\usepackage{boxedminipage}
\usepackage{xcolor}
\usepackage[linesnumbered,boxed,algosection]{algorithm2e}
\usepackage{framed}

\newcommand{\NP}{\text{\normalfont  NP}}

\newcommand{\FPT}{\textsf{FPT}}

\usepackage{amsmath, amssymb, latexsym}
\usepackage{enumerate}

\usepackage{float}
\usepackage{xspace}

\newcommand{\labelcoverv}{{\sc Node Unique Label Cover}}

\newcommand{\labelcovere}{{\sc Edge Unique Label Cover}}

\newcommand{\No}{{\sc No}}
\newcommand{\Yes}{{\sc Yes}}
\newcommand{\Oh}{{\mathcal{O}}}

\theoremstyle{plain}
\newtheorem{thm}{\bf Theorem}[section]
\newtheorem{lem}{\bf Lemma}[section]

\newtheorem{obs}{\bf Observation}[section]

\usepackage[ pdftex, plainpages = false, pdfpagelabels, 
                 bookmarks=false,
                 bookmarksopen = true,
                 bookmarksnumbered = true,
                 breaklinks = true,
                 linktocpage,
                 pagebackref,
                 colorlinks = true,  
                 linkcolor = blue,
                 urlcolor  = blue,
                 citecolor = red,
                 anchorcolor = green,
                 hyperindex = true,
                 hyperfigures
                 ]{hyperref}

\begin{document}

\title{A Linear Time Parameterized Algorithm for {\sc Node Unique Label Cover}}

 \author{
 Daniel Lokshtanov\thanks{University of Bergen, Bergen, Norway.\texttt{daniello@ii.uib.no}}
 \and M. S. Ramanujan\thanks{Algorithms and Complexity Group, TU Wien, Vienna, Austria     \texttt{ramanujan@ac.tuwien.ac.at}} 
 \and  Saket Saurabh\addtocounter{footnote}{-1}\footnotemark\thanks{The Institute of Mathematical Sciences, Chennai, India \texttt{saket@imsc.res.in}}
 }

\maketitle

\begin{abstract}
The optimization version of the {\sc Unique Label Cover} problem is at the heart of the Unique Games Conjecture which has played an important role in the proof of several tight inapproximability results. In recent years, this problem has been also studied extensively 
from the point of view of parameterized complexity. Cygan et al.~[FOCS 2012] proved that this problem is fixed-parameter tractable ({\FPT}) and Wahlstr\"{o}m [SODA 2014] gave an {\FPT} algorithm with an improved parameter dependence. Subsequently, Iwata, Wahlstr\"{o}m and Yoshida [2014] proved that the \emph{edge} version of {\sc Unique Label Cover} can be solved in \emph{linear} {\FPT}-time. That is, there is an {\FPT} algorithm whose dependence on the input-size is linear. However, such an algorithm for the \emph{node} version of the problem was left as an open problem. In this paper, we resolve this question by presenting the first linear-time {\FPT} algorithm for 
 {\labelcoverv}. 
\end{abstract}

\section{Introduction}


In the {\sc Unique Label Cover} problem we are given an undirected
graph $G$, where each edge $uv = e \in E(G)$ is associated with a permutation  $\phi_{e,u}$ of a constant size alphabet $\Sigma$. The goal is to construct a labeling  $\Psi:V(G) \to \Sigma$ maximizing the number of satisfied edge constraints, that is, edges for which $( \Psi(u),\Psi(v))\in \phi_{uv,u}$ holds.  For some $\epsilon >0$ and given  {\sc Unique Label Cover} instance $L$, 
{\sc Unique Label Cover}$(\epsilon)$ is the decision problem of distinguishing between the following two cases: (a) there is a labeling $\Psi$ under which at least $(1-\epsilon) |E(G)|$ edges are satisfied; and (b) for every labeling  $\Psi$ at most $\epsilon |E(G)|$ edges are satisfied.  This problem is at the heart of famous Unique Games Conjecture (UGC) of Khot~\cite{Khot02a}. Essentially, UGC says that for  any 
$\epsilon > 0$, there is a constant $M$ such that it is NP-hard to decide {\sc Unique Label Cover}$(\epsilon)$ on instances with label set of size $M$. The {\sc Unique Label Cover}($\epsilon$)  problem over the years has become  a canonical problem to obtain tight inapproximbaility results.  
We refer the reader to a survey of Khot~\cite{Khot10} for more detailed discussion on UGC.

In recent times {\sc Unique Label Cover} has also attracted a lot of attention in the realm of parameterized complexity. In particular two parameterizations, namely, {\labelcovere} and {\labelcoverv} have been extensively studied. These problems are, not only, interesting combinatorial problems on its own but they also generalize several well-studied problems in the realm of parameterized complexity.  The objective of this paper is to study the following problem.  

\defparproblem{{\labelcoverv}}
{A simple graph $G$, finite alphabet $\Sigma$, integer $k$ and for every edge $e$ and each of its endpoints $u$, a permutation $\phi_{e,u}$ of $\Sigma$ such that if $e=(u,v)$ then $\phi_{e,u}=\phi_{e,v}\inv$ and a function $\tau:V\to 2^\Sigma$.}
{$\vert \Sigma\vert+k$}{Is there a set $X\subseteq V(G)$ and a function $\Psi:V(G)\setminus X \to \Sigma$ such that for any $v\in V(G)\setminus X$, $\Psi(v)\in \tau(v)$ and for any $(u,v)\in E(G-X)$, we have $( \Psi(u),\Psi(v))\in \phi_{uv,u}$?
}

\medskip
We remark that the standard formulation of this problem excludes the function $\tau$. However, this formulation is a clear generalization of the standard formulation (simply set $\tau(v)=\Sigma$ for every vertex $v$) and the way we describe our algorithm makes it notationally convenient to deal with this statement. To make the presentation simpler, we assume that $\Sigma = [|\Sigma|]=\{1,\dots, |\Sigma|\}$. The {\labelcoverv} generalizes a well-studied graph problem, namely  the 
 {\sc Group Feedback Vertex Set} problem ~\cite{Guillemot11a} and thus 
{\sc Odd Cycle Transversal}, {\sc Feedback Vertex Set}.  It also encompasses the 
 {\sc Multiway Cut} problem.

%


The parameterized complexity of the {\labelcoverv} problem was first studied by Chitnis et al.~\cite{ChitnisCHPP12} who proved it is {\FPT} by giving an algorithm running in time $2^{\Oh(k^2\cdot \log |\Sigma|)}n^4 \log n$. Subsequently, Wahlstr\"{o}m \cite{Wahlstrom14} improved the parameter dependence by giving an algorithm running in time $\Oh(\vert \Sigma \vert ^{2k} n^{\Oh(1)})$. The \emph{edge} version of this problem was proved to be solvable in {\FPT}-\emph{linear} time by Iwata et al.~\cite{IwataWY15} who gave an algorithm running in time $\Oh(\vert \Sigma \vert ^{2k}(m+ n))$.  However, their approach does not apply to the much more general \emph{node} version of the problem and the existence of an {\FPT} algorithm with a linear time dependence on the input size has remained an open question. In this paper, we answer this question in the affirmative by giving a linear time {\FPT} algorithm for this problem. 
Note that we have stated the problem in a slightly more general form than is usually seen in literature. However, this modification does not affect the solvability of the problem in linear {\FPT} time.
We now state our theorem formally.

\begin{thm}\label{thm:main_theorem}
	There is a $2^{\Oh(k\cdot |\Sigma| \log |\Sigma|)}(m+n)$
	 algorithm solving {\labelcoverv}, where $m$ and $n$ are the number of edges and vertices respectively in the input graph.
\end{thm}

This answers an open question of Iwata et al.~\cite{IwataWY15}.
Furthermore, when the label set $\Sigma$ is of constant-size for some fixed constant, our algorithm achieves optimal asymptotic dependence on the budget $k$ under the Exponential Time Hypothesis \cite{ImpagliazzoPaturiZane01}.

%

By its very nature, the {\labelcoverv} problem is a problem about breaking various types of dependencies between vertices. Since these dependencies are propagated along edges, it is reasonable to view the problem as breaking these dependencies by hitting appropriate sets of paths in the graph. Chitnis et al.~\cite{ChitnisCHPP12} used this idea to argue that highly connected pairs of vertices will always remain dependent on each other and hence one can recursively solve the problem by first designing an algorithm for graphs that are `nearly' highly connected and then use this algorithm as a base case in a divide and conquer type approach. However, the polynomial dependence of their algorithm is $\Oh(n^4 \log n)$ where $n$ is the number of vertices in the input. Subsequently, 
Wahlstrom \cite{Wahlstrom14} improved the parameter dependence by using a branching algorithm based on the solution to a specific linear program. However, since this algorithm requires solving linear programs, the dependence on the input is far from linear. Iwata et al.~\cite{IwataWY15} showed that for several special kinds of LP-relaxations, including those involved in the solution of the \emph{edge} version of {\sc Unique Label Cover}, the corresponding linear program can be solved in linear-time using flow-based techniques and hence they were able to obtain the first linear-time {\FPT} algorithm for the edge version of {\sc Unique Label Cover}. However, their approach fails when it comes to the \emph{node} version of this problem.

\medskip
\noindent
\textbf{Our Techniques.}
In this paper, we view the {\labelcoverv} problem as a problem of hitting paths between certain pairs of vertices in an appropriately designed \emph{auxiliary} graph $H$ whose size  is greater than that of the input graph $G$ by a factor depending only on the parameter. 
We then show that for any prescribed labelling on the vertices of $G$, 
 it is possible to  select (in linear time)  a constant-size set of vertices of $G$ such that after guessing the intersection of this set with a hypothetical solution, if we augment the labelling by branching over all permitted labellings of the remaining vertices in this set then we reduce a pre-determined measure of the input which depends only on the parameter. By repeatedly doing this, we obtain a branching algorithm for this problem where each step requires linear time.
The main technical content of the paper is in proving that :
   
  \begin{itemize} 
  \item there exists a constant-size vertex set and an appropriate measure for the instance such that the measure `improves' in each step of the branching and
   \item such a vertex set can be computed in linear time.
  
  \end{itemize}

\smallskip
\noindent 
{\bf Related work on improving dependence on input size in {\FPT} algorithms.}  Our algorithm for {\labelcoverv}  belongs to a large body of work where the main goal is to design linear time algorithms  for \NP-hard problems for a fixed value of $k$. That is, to design  an algorithm with running time $f(k)\cdot \Oh(|I|)$, where $|I|$ denotes the size of the input instance.  This area of research predates even parameterized complexity. The genesis of parameterized complexity is in the theory of graph minors, developed by Robertson and Seymour~\cite{RobertsonS95b,RobertsonS03b,RobertsonS04}. Some of the important algorithmic consequences of this theory include $\Oh(n^3)$ algorithms for {\sc Disjoint Paths} and {\sc $\cal F$-Deletion} for every fixed values of $k$. 
These results led to a whole new area of designing algorithms for \NP-hard problems with as small dependence on the input size as possible; resulting in 
algorithms with improved dependence on the input size for {\sc Treewidth}~\cite{stocBodlaender93,Bodlaender96}, \FPT{} approximation for 
{\sc Treewidth}~\cite{BodlaenderDDFLP13,Reed92}
{\sc Planar $\cal F$-Deletion}~\cite{stocBodlaender93,Bodlaender96,FellowsL88,FominLMS12,FominLMRS15}, and {\sc Crossing Number}~\cite{stocGrohe01,Grohe04,KawarabayashiR07}, to name a few.

The advent of parameterized complexity started to shift the focus away from the running time dependence on input size to the dependence on the parameter. That is, the goal became designing parameterized algorithms with running time upper bounded by $f(k)n^{\Oh(1)}$, where the function $f$ grows as slowly as possible.
Over the last two decades researchers have tried to optimize one of these objectives, but rarely both at the same time. More recently, efforts have been made towards obtaining linear (or polynomial) time parameterized algorithms that compromise as little as possible on the dependence of the running time on the parameter $k$. The gold standard for these results are algorithms with linear dependence on input size as well as provably optimal (under ETH) dependence on the parameter. New results in this direction include parameterized algorithms for problems such as {\sc Odd Cycle Transversal}~\cite{IwataOY14,RamanujanS14}, {\sc Subgraph Isomorphism}~\cite{Dorn10}, {\sc Planarization}~\cite{JansenLS14,Kawarabayashi09}, {\sc Subset Feedback Vertex Set}~\cite{LokshtanovRS15}
as well as a single-exponential and linear time parameterized constant factor approximation algorithm for {\sc Treewidth}~\cite{BodlaenderDDFLP13}. Other recent results include parameterized algorithms with improved dependence on input size for a host of problems~\cite{GroheKR13,KawarabayashiKR12,KawarabayashiM08,KawarabayashiMR08,KawarabayashiR09,KawarabayashiR10}. 
\medskip

\smallskip
\noindent 
{\bf Related work on graph separation in {\FPT} algorithms.}  Marx~\cite{Marx06} was the first to consider cut problems 
in the context of parameterized complexity. He observed that the {\sc Multiway Cut} problem can be 
shown to be {\FPT} by a simple application of graph minors, (see~\cite[Section $3$]{Marx06}) and then went on to give an algorithm for the same problem 
with a running time of $\Oh(4^{k^3}n^{\Oh(1)})$. 
The notion of \emph{important separators} which was introduced in this paper has been instrumental in settling the parameterized complexity of numerous graph-separation problems including
{\sc Directed Feedback vertex Set}~\cite{ChenLLOR08}, {\sc Almost $2$ SAT}~\cite{RazgonO09}, {\sc Multicut} \cite{MarxR14}, the directed versions of {\sc Multiway Cut} \cite{ChitnisHM13}, {\sc Subset Feedback Vertex Set} \cite{ChitnisCHM15}, {\sc Multicut} restricted to acyclic digraphs \cite{KratschPPW15} as well as parity based generalizations of {\sc Multiway Cut} \cite{LokshtanovR12}. 

%

%
%

\section{Preliminaries}

We fix a label set $\Sigma$ and assume that all instances of {\labelcoverv} we deal with are over this label set.
When we refer to a set $X$ being a \emph{solution} for a given instance of {\labelcoverv}, we implicitly assume that $X$ is a set of \emph{minimum} size. We denote the set of functions $\{\phi_{e,u}\}_{e\in E(G),u\in e}$ simply as $\phi$ (without any subscript).

Before we proceed to describe our algorithm for {\labelcoverv}, we make a few remarks regarding the 
representation of the input.
 We assume that the input graph is given in the form of an adjacency list and for every edge $e=(u,v)$ the permutations $\phi_{e,v}$ and $\phi_{e,u}$ are included in the two nodes of the adjacency list corresponding to the edge $e$. This is achieved by representing the permutations as  $\vert \Sigma\vert$-length arrays over the elements in $[|\Sigma|]$. 
 It is straightforward to check that given the input to {\sc Label Cover} in this form, the decision version of the problem can be solved in time $\Oh(|\Sigma|^{\Oh(1)}(m+n))$. We assume that the input to {\labelcoverv} is also given in the same manner. 

\section{Setting up the tools}



\subsection{Defining the auxiliary graph}

\begin{definition}
	
	Let $(G,k,\phi,\tau)$ be an instance of {\labelcoverv} and let $\Psi:V\to \Sigma$. We say that $\Psi$ is a \textbf{feasible labeling} for this instance if for all $(u,v)\in E(G)$, $(\Psi(u),\Psi(v))\in \phi_{uv,u}$. For $\tau:V\to 2^\Sigma$, we say that $\Psi$ is \textbf{consistent with} $\tau$ if for every $v\in V(G)$, $\Psi(v)\in \tau(v)$.
\end{definition}

For an instance $I=(G,k,\phi,\tau)$ of {\labelcoverv}, we define an associated auxiliary graph $H_I$ as follows. The vertex set of $H_I$ is $V(G) \times \Sigma$. For notational convenience, we denote the vertex $(v,i)$ by $v_i$. The vertex $v_i$ is meant to represent the (eventual) labeling of $v$ by the label $i$. The edge set of $H_I$ is defined as follows. For every edge $e=(u,v)$ and for every $i\in \Sigma$, we have an edge $(u_i,v_{\phi_{e,u}(i)})$. That is, we add an edge between $u_i$ and $u_j$ where $j$ is the image of $i$ under the permutation $\phi_{e,u}$. 

We now prove certain structural lemmas regarding this auxiliary graph which will be used in the design as well as analysis of our algorithm. For ease of description, we will treat instances of {\sc Label Cover} as instances of {\labelcoverv}. To be precise, we represent an instance $(G,\phi)$ of {\sc Label Cover} as the trivially equivalent instance $(G,0,\phi,\tau^0)$ of {\labelcoverv} where, $\tau^0(v)=\Sigma$ for every $v\in V(G)$. The first observation follows from the definition of $H_I$ and the fact that since $G$ is a simple graph, for every edge $e\in E(G)$, the set of edges in $H_I$ that correspond to this edge form a matching.

  \begin{obs} \label{obs:nocommonneighbors}
	   Let $I=(G,0,\phi,\tau)$ be an instance of {\labelcoverv}. Then, for every $v\in V(G)$, for every distinct $i,j\in \Sigma$, $v_i$ and $v_j$ have no common neighbors in $H_I$. 
	   \end{obs}

\begin{obs}\label{obs:color_propagation}
Let $I=(G,0,\phi,\tau)$ be a {\Yes} instance of {\labelcoverv} and let $\Psi$ be a feasible labeling for this instance. Let $v\in V(G)$ and $i=\Psi(v)$. Then, for any vertex $u\in V(G)$ and $j\in \Sigma$, if $u_j$ is in the same connected component as $v_i$ in $H_I$ then $\Psi(u)=j$. 
\end{obs}

\begin{proof}
	The proof is by induction on the length of a shortest path in $H_I$ between $v_i$ and $u_j$. In the base case, suppose that $v_i$ and $u_j$ are adjacent. Then, by the definition of $H_I$, it must be the case that $(u,v)$ is an edge in $G$ and furthermore, $j=\phi_{uv,u}(i)$. Since $\Psi$, is feasible, it 
	   follows that $\Psi(u)=j$. We now move to the induction step and suppose that $P$ is a shortest path in $H_I$ from $v_i$ to $u_j$, where the length of $P$ is at least 2. Let $w\in V(G)$ and $r\in \Sigma$ such that $(w_r,u_j)$ is the last edge encountered when traversing $P$ from $v_i$ to $u_j$. Then, by the induction hypothesis, we can conclude that $\Psi(w)=r$. Furthermore, by the definition of $H_I$, it must be the case that $(w,u)$ is an edge in $G$ and $j=\phi_{wu,w}(r)$. Therefore, the feasibility of $\Psi$ implies that $\Psi(u)=j$. This completes the proof of the observation.
	   \end{proof}

	   The above observation describes the `dependency' between pairs of vertices which are in the same connected component of $G$. Moving forward, we will characterize the dependencies between vertices when subjected to additional constraints. Before we do so, we need the following definitions.

	   \begin{definition}
	   Let $I=(G,k,\phi,\tau)$ be an instance of {\labelcoverv}.
	For $v\in V(G)$, we use $[v]$ to denote the set $\{v_1,\dots, v_{|\Sigma|}\}$. For a subset $S\subseteq V(G)$, we use $[S]$ to denote the set $\bigcup_{v\in S}[v]$.  Similarly, for $e=(u,v)\in E(G)$, we use $[e]$ to denote the set $\mathlarger{\mathlarger{\{}}(u_i,v_j)\mathlarger{\mathlarger{\}}}_{i\in \Sigma, j=\phi_{e,u}(i)}$ of edges and for a subset $X\subseteq E(G)$, we use $[X]$ to denote the set $\bigcup_{e\in X}[e]$.  For the sake of convenience, we also reuse the same notation in the following way. For $v\in V(G)$ and $\alpha\in \Sigma$, we also use $[v_\alpha]$ to denote the set $\{v_1,\dots, v_{|\Sigma|}\}$. This definition  extends in a natural way to sets of vertices and edges of the auxiliary graph $H_I$. Finally, for a set $S\subseteq V(H_I)\cup E(H_I)$, we denote by $S\inv$ the set $\{s|s\in V(G)\cup E(G): [s]\cap S\neq \emptyset\}$.
	\end{definition}

	   \begin{definition}
	      Let $I=(G,k,\phi,\tau)$ be an instance of {\labelcoverv}.
	We say that a set $Z\subseteq V(H_I)\cup E(H_I)$ is \textbf{regular} if $|Z\cap [v]|\leq 1$ for any $v\in V(G)$ and $|Z\cap [e]|\leq 1$ for any $e\in V(G)$ and \textbf{irregular} otherwise. That is, regular sets contain at most 1 copy of any vertex and edge of $G$. 
\end{definition}

Now that we have defined the notion of regularity of sets, we prove the following lemma which shows that the auxiliary graph displays a certain symmetry with respect to regular paths. This will allow us to transfer arguments which involve a regular path between vertices $v_i$ and $u_j$ to one between vertices $v_{i_1}$ and $u_{j_1}$ where $i\neq i_1$ and $j\neq j_1$.

\begin{lem}\label{lem:path_copies}
	Let $I=(G,k,\phi,\tau)$ be an instance of {\labelcoverv}. Let $P$ be a regular path in $H_I$ from $v_i$ to $u_j$. Let $V(P)$ denote the set of vertices of $G$ in $P$ and let $U$ denote the set $[V(P)]$. Then, there are vertex disjoint paths $P_1,\dots, P_{|\Sigma|}$ in $H_I$ and a partition of $U$ into sets $U_1,\dots, U_{|\Sigma|}$ such that for each $r\in [|\Sigma|]$, $V(P_r)=U_r$ and $P_r$ is a path from $v_{i_1}$ to $u_{i_2}$ for some $i_1,i_2\in \Sigma$.
\end{lem}

%
%
%

\begin{proof}
The proof is by induction on the length of $P$. In the base case, suppose that $P$ is a single edge which corresponds to the edge $e\in E(G)$. That is, $P=(v_i,u_j)\in E(H_I)$ and $U=[\{v,u\}]$. For each $r\in \Sigma$, we  define $P_r$ to be the edge $(v_r,u_{\phi_{e,v}(r)})$ and $U_r$ to be the set $\{v_r,u_{\phi_{e,v}(r)}\}$. Observe that the statement of the lemma holds with respect to these sets. We now move to the induction step, where $P$ has length at least 2. Let $s\in \Sigma$ and $w\in V(G)$ such that $(w_s,u_j)$ is the last edge of $P$ encountered when traversing $P$ from $v_i$ to $u_j$. We now apply the induction hypothesis on the subpath of $P$ from $v_i$ to $w_s$ and the above argument for the base case on the subpath of $P$ from $w_s$ to $u_j$ which is precisely the edge $(w_s,u_j)$. Let the first subpath be $Q$ and the second subpath $J$.

Let $Q_1,\dots,Q_{|\Sigma|}$ be the paths and $U^Q_1,\dots, U^Q_{|\Sigma|}$ be the partition of $[V(Q)]$ given by the induction hypothesis. Similarly, let $J_1,\dots, J_{|\Sigma|}$ be the paths and $U^J_1,\dots, U^J_{|\Sigma|}$ be the partition of $[V(J)]$ given by our arguments for the base case. Since these are partitions and $[V(Q)]$ and $[V(J)]$ intersect in precisely the set $[w]$, 
 we may assume without loss of generality that for every $r\in \Sigma$, the sets $U^Q_r$ and $U^J_r$ contain the vertex $w_r$.  For each $r\in \Sigma$, we now define $U_r$ to be $U^Q_r\cup U^J_r$ and $P_r$ to be the concatenated path $Q_r \oplus J_r$. Since $Q_r$ is a path with $w_r$ as one endpoint and $J_r$ is a path (indeed an edge) with $w_r$ as an endpoint for each $r\in \Sigma$, the path $P_r$ is well-defined. Further, since the sets $U^Q_1,\dots, U^Q_{|\Sigma|}$ partition the set $[V(Q)]$ and $U^J_1,\dots, U^J_{|\Sigma|}$ partition $[V(J)]$, we conclude that $U_1,\dots, U_{|\Sigma|}$ indeed partition $[V(P)]$ and for each $r$, $V(P_r)=U_r$. 
This completes the proof of the lemma.	
\end{proof}


In the next lemma, we describe additional structural properties of the auxiliary graph. In particular, we establish the relation between various copies of the same vertex set. Intuitively, the following lemma says that for every connected and regular set of vertices $Z$, simply observing the set $N[Z]$ can allow one to make certain useful assertions about the set of vertices in the neighborhood of the set $Z'=[Z]\setminus Z$.

\begin{lem}\label{lem:partial_symmetry}
Let $Z\subseteq V(H_I)$ be a connected regular set of vertices and let $Y=N(Z)$. 
Further, suppose that $N[Z]$ is regular.
 Let $Z'=[Z]\setminus Z$ and $Y'=[Y]\setminus Y$. Then, $Y'\subseteq N(Z')\subseteq [Y]$. Furthermore, for every connected component $C$ in $H_I[Z']$, $N(C)\cap [v]\neq \emptyset$ for every $v\in V(G)$ for which there is a $j\in \Sigma$ such that $v_j\in Y$.
\end{lem}

\begin{proof}
	We begin by arguing that $Y'\subseteq N(Z')$. That is, for every vertex $a\in Y'$, there is a vertex $b\in Z'$ such that $(a,b)\in E(H_I)$. Consider a vertex $a=x_i\in Y'$ where $x\in V(G)$ and $i\in \Sigma$. By the definition of $Y'$, there is a $j\in \Sigma$ such that $i\neq j$ and $x_j\in Y$. Since $Y=N(Z)$, it must be the case that there is a $y\in V(G)$ and $r\in \Sigma$ such that $(y_r,x_j)\in E(H_I)$ and $y_r\in Z$. Now, by the definition of $Z'$, we know that for every $s\in \Sigma \setminus \{r\}$, the vertex $y_s\in Z'$. Furthermore, the presence of the edge $(y_r,x_j)\in E(H_I)$ implies the presence of an edge $(x_i,y_\ell)\in E(H_I)$ for some $\ell\in \Sigma$. From Observation \ref{obs:nocommonneighbors}, we infer that $\ell\neq r$. Since we have already argued that $y_\ell\in Z'$, we conclude that $x_i\in N(Z')$.

  We now argue that $N(Z')\subseteq [Y]$. For this, we need to show that for every $a\in Z'$ and $b\notin Z'$ such that $(a,b)\in E(H_I)$, it must be the case that $b\in [Y]$.  Consider a vertex $x_i\in Z'$ where $x\in V(G)$ and $i\in \Sigma$ and a vertex $y_r\notin Z'$ for some $y\in V(G)$ and $r\in \Sigma$ such that $(x_i,y_r)\in E(H_I)$. By the definition of $Z'$, there is a $j\in \Sigma$ such that $i\neq j$ and $x_j\in Z$. Now, due to the edge $(x_i,y_r)$, we have the existence of the edge $(x_j,y_s)\in E(H_I)$ for some $s\in \Sigma$. Due to Observation \ref{obs:nocommonneighbors}, we know that $s\neq r$ since $y_r$ is already neighbor to $x_i$. Therefore, if $y_s\in Z$, then $y_r$ would be in $Z'$, a contradiction. This allows us to infer that $y_s\notin Z$, implying that $y_s\in N(Z)=Y$. But this means that $y_r\in [Y]$, completing the proof of this statement as well.
  
  Finally, we address the last statement of the lemma. That is, the neighborhood of each connected component induced by the set $Z'$ contains at least one copy of every vertex of $Y$.  For this, we require the following claim.

     	\begin{claim}
     	For any vertex $t\in V(G)$ and label $\alpha\in \Sigma$, if there is a component $C$ of $H_I[Z']$ containing $t_\alpha$, then there is a label $\beta\in \Sigma$ such that $t_\beta\in Z$ and furthermore, for every $p\in V(G)$ and $\gamma\in \Sigma$, if $p_\gamma\in Z$ then there is a $\delta\in \Sigma$ such that $p_\delta$ is in $C$.
     \end{claim}

\begin{proof}

It follows from the definition of $Z'$ that if $t_\alpha$ is in $Z'$, then there must be a label $\beta\neq \alpha$ such that $t_\beta\in Z$. Now, suppose that $p_\gamma\in Z$. Since $Z$ is connected, there is a path from $t_\beta$ to $p_\gamma$ contained within $Z$. We prove the statement of the claim by induction on the length of a shortest path between these 2 vertices which is contained within $Z$.

 Let $P$ be such a shortest path and in the base case, suppose that $P$ is an edge. That is, $(t_\beta,p_\gamma)\in E(H_I)$. Then, the definition of $H_I$ implies the existence of a $\delta$ such that $(t_\alpha,p_\delta)\in E(H_I)$. Furthermore, by Observation \ref{obs:nocommonneighbors}, we know that $\delta\neq \gamma$. Since $C$ is a connected component of $Z'$ and both $t_\alpha$ and $p_\delta$ are in $Z'$, we conclude that $p_\delta$ is in $C$. We now perform the induction step where $P$ is a path of length at least 2, assuming our statement holds for all paths of length at most $|P|-1$. 
 
        Let $w\in V(G)$ and $r\in \Sigma$ such that the last edge on this path when traversing from $t_\beta$ to $p_\gamma$ is the edge $(w_r,p_\gamma)$. Then, by the induction hypothesis, there is an $\ell\in \Sigma$ such that $w_\ell$ is in $C$. Now, invoking the same argument as above, we infer the existence of a $\delta\in \Sigma$ such that $p_\delta$ is also in $C$, completing the proof of the claim.  
  \end{proof}

  	We now complete the proof of the final statement of the lemma.  Consider a connected component $C$ of the graph $H_I[Z']$ and consider the set $N(C)$. 
 Suppose that for some $v\in V(G)$ and $j\in \Sigma$, $v_j\in Y$. Consider a vertex $t_\alpha$ in $C$ where $t\in V(G)$ and $\alpha\in \Sigma$. Then, by the above claim, there is a $\beta\in \Sigma$ such that $t_\beta\in Z$. Since $Y=N(Z)$, we infer the existence of a $p_\gamma$ in $Z$ which is adjacent to $v_j$ in $H_I$. Invoking the above claim again, we infer the existence of a $\delta\in \Sigma$ such that $p_\delta$ is in $C$. However, by the definition of $H_I$, the presence of the edge $(v_j,p_\gamma)$ implies the presence of an edge $(v_i,p_\delta)$ for some $i\in \Sigma$. Observe that $v_i$ cannot be in $Z'$. This is because if $v_i\in Z'$, then $N[Z]$ is not regular,
  a contradiction to the premise of the lemma. 
 Therefore, it must be the case that $v_i\in N(C)$, implying that $N(C)\cap [v]\neq\emptyset$. This completes the proof of the lemma.	
\end{proof}

%

Using the observations and structural lemmas proved so far, we will now give a forbidden-structure characterization of  {\Yes} instances of {\labelcoverv}.

	   \begin{lem}\label{lem:characterization}
	   Let $I=(G,0,\phi,\tau)$ be a {\Yes} instance of {\labelcoverv} where $G$ is connected. Let $v\in V(G)$ and $i\in \Sigma$. Then, there is a feasible labeling $\Psi$ such that $\Psi(v)=i$  if and only if there is no $j\in \Sigma$ such that $v_i$ and $v_j$ are in the same connected component of $H_I$.
	  \end{lem}

\begin{proof}
  Consider the forward direction of the lemma. That is, the claim that if there is a feasible labeling $\Psi$ such that $\Psi(v)=i$ then, there is no $j$ such that $v_i$ and $v_j$ are in the same connected component of $H_I$. The contra-positive of this statement is the following. If there is a $j\in \Sigma$ such that $v_i$ and $v_j$ are in the same connected component of $H_I$ then there is no feasible labeling $\Psi$ such that $\Psi(v)=i$. By applying the statement of Observation \ref{obs:color_propagation} on the vertices $v_i$ and $v_j$, we conclude that the forward direction holds.

	We now argue the converse direction. That is, if there is no $j\in \Sigma$ such that $v_i$ and $v_j$ are connected in $H_I$, then there is a feasible labeling $\Psi$ such that $\Psi(v)=i$. 
	Observe that since $G$ is connected, for every vertex $u\in V(G)$, there is at least one vertex of $[u]$ in the same component as $v_i$ in $H_I$, call it $C$. 
	We now consider 2 cases. In the first case, $|[u]\cap C|=1$ for every $u\in V(G)$. In the second case, there is a vertex $u\in V(G)$ such that $|[u]\cap C|>1$.

	\begin{description}
		\item[Case 1:] 
		We  define the labeling $\Psi$ on $V(G)$ as follows. For every $u\in V(G)$, let $\Psi(u)=r$ if and only if $u_r\in C$. Clearly, $\Psi$ is well-defined. We claim that $\Psi$ is in fact a feasible labeling with $\Psi(v)=i$. It is clear from the definition of $\Psi$ that $\Psi(v)=i$. Hence it only remains to prove that $\Psi$ is feasible. Suppose that $\Psi$ is not feasible and let $e=(p,q)\in E(G)$ such that $\Psi(p)=\alpha$ and $\Psi(q)=\beta$ do not match. That is, $\beta\neq \phi_{e,p}(\alpha)$. By the definition of $\Psi$, $p_\alpha$ and $q_\beta$ are both in $C$. Furthermore, by the definition of $H_I$, $(p_\alpha,q_\beta)\notin E(H_I)$ and there is a $\gamma \neq \beta$ such that $(p_\alpha,q_\gamma)\in E(H_I)$. However, this implies that $q_\beta,q_\gamma \in C$, a contradiction to our assumption that $|C\cap [q]|=1$. Therefore, we conclude that $\Psi$ is indeed a feasible labeling setting $\Psi(v)=i$, completing the argument for this case.

		\item[Case 2:] In this case, there is a $u\in V(G)$ such that $|[u]\cap C|>1$. Let $\alpha,\beta\in \Sigma$ be such that $u_\alpha,u_\beta\in C$. Consider a path $P$ from $v_i$ to $u_\alpha$. Let $P_1$ be a regular subpath of $P$ with endpoints $v_i$ and $w_\gamma$ where $|[w]\cap C|>1$ and we choose $\delta\in \Sigma$ such that $w_\delta\in C$. If $P$ is already regular then $P_1=P$, $w=u$, $\gamma=\alpha$ and $\delta=\beta$. However, if $P$ is not regular, then choose $w_\gamma$ to be the vertex on $P$ closest to $v_i$ such that $|[w]\cap C|>1$ and $w_\delta$ to be another vertex in $[w]$ which lies on $P$. Since by the premise of the statement, for no $j\in \Sigma$ is the vertex $v_j\in C$, the path $P$ contains at least 1 edge.
		  Now, we apply Lemma \ref{lem:path_copies} on the regular path $P_1$ and obtain regular paths $P'_1,\dots, P'_{|\Sigma|}$ with each path having as one of its endpoints a unique vertex from $[v]$ and the other endpoint a unique vertex from $[w]$. Since one of these paths is $P_1$ itself, we assume without loss of generality that $P_1'=P_1$, which is a path from $v_i$ to $w_\gamma$ and $P_2'$ is a path from $v_j$ to $w_\delta$ for some $j\in \Sigma$ where $j\neq i$. Since $w_\gamma$ and $w_\delta$ are in $C$, we infer that $v_i$ and $v_j$ are also in $C$, contradicting the premise.

	\end{description}

	This completes the argument for the second case as well and hence the proof of the converse direction of the lemma.
%
%
%
%
%
%
%
%
%
%
%
%
%
%
%
%
%
%
%
%
%
%
%
\end{proof}

	  In the next lemma,  we extend the statement of the previous lemma to include a description of {\Yes} instances where $k=0$ and the given graph has a feasible labeling that is consistent with a given function $\tau$.

\begin{lem}\label{lem:main_characterization}
	
Let $I=(G,0,\phi,\tau)$ be an instance of {\labelcoverv}. Then, $I$ is a {\Yes} instance if and only if 
for every vertex $v\in V(G)$, there is an $i\in \Sigma$ such that there is no path in $H_I$ from $v_i$ to $v_j$ for any $j\neq i$. Moreover, if there is a feasible labeling $\Psi$ for $G$ consistent with $\tau$ such that $\Psi(v)=i$ then there is no vertex $u\in V(G)$ and label $j\in \Sigma\setminus\tau(u)$ such that there is a path in $H_I$ from $v_i$ to $u_j$.
\end{lem}

\begin{proof}
	Suppose that $I$ is a {\Yes} instance and let $\Psi$ be a  feasible labeling of $G$. Let $v\in V(G)$ and let $i=\Psi(v)$. Then, Observation \ref{obs:color_propagation} implies that $H_I$ contains no $v_i$-$v_j$ path for any $j\neq i$.

	Conversely, suppose that for every vertex $v\in V(G)$, there is a $i_v\in \Sigma$  such that there is no path in $H_I$ from $v_{i_v}$ to $v_j$ for any $\Sigma\ni j\neq i_v$. For each connected component $C$ of $G$, 
	    we pick an arbitrary vertex $w$ and apply Lemma \ref{lem:characterization} to conclude that there is a feasible labeling $\Psi_C$ for this component which sets $\Psi_C(w)=i_w$. Finally the labeling $\Psi$ defined as the union of the labelings $\{\Psi_C\}_{C\in {Comp}(G)}$ is a feasible labeling for $G$, completing the proof of the first statement of the lemma.

	    Now, suppose that $\Psi$ is a feasible labeling of $G$ consistent with $\tau$ and $\Psi(v)=i$. Suppose that $v_i$ is in the same component as $u_j$ where $u\in V(G)$ and $j\in \Sigma\setminus \tau(u)$. Then, Observation \ref{obs:color_propagation} implies that $\Psi(u)=j\notin \tau(u)$, a contradiction to our assumption that $\Psi$ is  consistent with $\tau$. This completes the proof of the lemma.
	\end{proof}

So far, we have studied the structure of {\Yes} instances of this problem when the budget $k=0$.
The next lemma is a direct consequence of Lemma \ref{lem:main_characterization} and allows us to characterize {\Yes} instances of the problem for values of $k$ greater than 0.

\begin{lem}\label{lem:separator_characterization}
	Let $I=(G,k,\phi,\tau)$ be an instance of {\labelcoverv}. Then, $I$ is a {\Yes} instance if and only if there is a set $S\subseteq V(G)$ of at most $k$ vertices such that for every $v\in V(G)\setminus S$, there is an $i_v\in \Sigma$ such that $[S]$ intersects all paths from $v_{i_v}$ to $v_j$ for every $\Sigma\ni j\neq i_v$ in the graph $H_I$. Moreover if there is a feasible labeling for $G-S$ consistent with $\tau$ that labels $v$ with the label $i\in \tau(v)$ then for every $u\in V(G)$ and $j\in \Sigma\setminus \tau(u)$, $[S]$ intersects all $v_i$-$u_j$ paths.
\end{lem}

Using the above lemma, we will interpret the {\labelcoverv} problem as a parameterized cut-problem and use separator machinery to design a linear-time {\FPT} algorithm for this problem.

\subsection{Defining the associated cut-problem}

We begin by recalling standard definitions of separators in undirected graphs.

\begin{definition}
	Let $G$ be a graph and $X$ and $Y$ be disjoint vertex sets. A set $S$ disjoint from $X\cup Y$ is said to be an $X$-$Y$ \textbf{separator} if there is no $X$-$Y$ path in the graph $G-S$. We denote the vertices in the components of $G-S$ which intersect $X$ by $R(X,S)$ and we denote by $R[X,S]$ the set $R(X,S)\cup S$. We say that an $X$-$Y$ separator $S_1$ \textbf{covers} an $X$-$Y$ separator $S_2$ if $R(X,S_1)\supseteq R(X,S_2)$.
\end{definition}

\begin{definition} Let $I$ be an instance of {\labelcoverv} and let $X$ and $Y$ be disjoint vertex sets of $H_I$.
		We say that a minimal $X$-$Y$ separator $S$ is \textbf{good} if the set $R[X,S]$ is regular and \textbf{bad} otherwise.
\end{definition}

Note that if $S$ is a minimal $X$-$Y$ separator then $N(R(X,S))=S$.
We are now ready to prove the \emph{Persistence} Lemma which plays a major role in the design of the algorithm. In essence this lemma says that if we are guaranteed the existence of a solution whose deletion leaves a graph with a feasible labeling $\Psi$ and if we are given a vertex $v$ excluded from the deletion set which has a single label $\alpha$ in its allowed label set, then we can define a set $T$ such that the solution under consideration  \emph{must} separate $v_\alpha$ from $T$. Furthermore, if we find a \emph{good} minimum $v_\alpha$-$T$ separator $S$, then we can correctly fix the labels of all vertices which have exactly one copy in $R(v_\alpha,S)$. It will be shown later that once we fix the labels of these vertices, the subsequent exhaustive branching steps will decrease a pre-determined measure of the input instance.

\begin{lem}\label{lem:reduction}{\rm [{\bf Persistence Lemma}]}
	Let $I=(G,k,\phi,\tau)$ be a {\Yes} instance of {\labelcoverv}. Let $X\subseteq V(G)$ be a minimal set of size at most $k$ such that $G-X$ has a feasible labeling and let $\Psi$ be a feasible labeling for $G-X$ consistent with $\tau$. Let $v$ be a vertex not in $X$ with $|\tau(v)|=1$ and let $\alpha\in \Sigma$ be such that $\alpha=\Psi(v)$ and $\tau(v)=\{\alpha\}$. Let $T$ denote the set $\bigcup_{u\in V(G)} \bigcup_{\gamma\in \Sigma\setminus \tau(u)} u_\gamma $.

	\begin{itemize} 
	\item $[X]$ is a $v_\alpha$-$T$ separator in $H_I$. 
	\item Let $S$ be a good $v_\alpha$-$T$ minimum separator in $H_I$ and let $Z=R(v_\alpha,S)$.
	    Then, there is a solution for the given instance disjoint from $Z{\inv}$.

	\end{itemize}

\end{lem}

\begin{proof} The first statement follows from lemma \ref{lem:separator_characterization}. We now prove the second statement. We begin by observing that $T$ contains the set $[v]\setminus \{v_\alpha\}$. This is simply because $\tau(v)$ is a singleton and only contains the label $\alpha$. As a result, we know that the set $[X]$ must intersect all $v_\alpha$-$v_\beta$ paths for $\alpha\neq \beta$.
 Let $X_1$ denote the set $X\cap Z\inv$. If $X_1$ is empty then we are already done. Therefore, $X_1\neq \emptyset$. Let $S'$ denote the subset of $S\setminus [X]$ which is not reachable from $v_\alpha$ in the graph $H_I-[X]$ via paths whose internal vertices lie in $Z$. We now  have 2 cases depending on $S'$ being empty or non-empty.  We will argue that the first case cannot occur since it contradicts the minimality of $X$. In the second case we use very similar arguments but show that we can modify $X$ to get an alternate solution $X'$ which is disjoint from the set $Z$.

\begin{description}

\item[Case 1: $S'$ is empty.] That is, every vertex in $S\setminus [X]$ is reachable from $v_\alpha$ in $H_I-[X]$ via paths whose internal vertices lie in $Z$.
Let $u\in X_1$ and let $b\in \Sigma$ such that $u_b\in Z$. Since $Z$ is regular, $Z\cap [u]$ must in fact be equal to $\{u_b\}$. We now claim that $X'=X\setminus \{u\}$ is also a set such that $G-X'$ has a feasible labeling, contradicting the minimality of $X$.


Suppose that this is not the case. That is, $G-X'$ does not have a feasible labeling. Since every connected component of $G-X'$ which does not contain $u$ is also a connected component of $G-X$, all such components do have a feasible labeling. Indeed any feasible labeling of $G-X$ restricted to the vertices in these components is a feasible labeling for these components. Therefore, there is a single component in $G-X'$ which does not have a feasible labeling --  the component containing $u$.

By Lemma \ref{lem:characterization}, if there is no $b'\in \Sigma\setminus \{b\} $ such that the connected component of $H_I-[X']$ containing $u_b$ also contains $u_{b'}$, then there is a feasible labeling of the component of $G-X'$ which contains $u$, a contradiction. Therefore, there is a $b'\in \Sigma\setminus \{b\}$ such that there is a $u_b$-$u_{b'}$ path in $H_I-[X']$. If this path contains vertices of $[u]$ other than $u_b$ and $u_{b'}$, then we pick the vertex of $[u]\setminus \{u_b\}$ which is closest to $u_b$ on this path and call it $u_{b'}$. Therefore, the path $P$ from $u_b$ to $u_{b'}$ is internally disjoint from $[u]$. 
We now have the following claim regarding $P$.
\medskip

\begin{claim}
The path $P$ is internally regular.	
\end{claim}

\begin{proof}
Suppose that the path $P$ contains a pair of internal vertices from the set $[l]$ for some vertex $l\in V(G)\setminus \{u\}$. Let these vertices be $l^1$ and $l^2$. We consider the following 2 cases. In the first case, both $l^1$ and $l^2$ are disjoint from $Z$ and in the second, exactly one of them, say $l^1$ is contained in $Z$. Since $Z$ is regular, these are the only 2 possible cases. 

   We begin with the first case. That is, both $l^1$ and $l^2$ are disjoint from $Z$. Since $u_b\in Z$ and $P$ contains $u_b$ and $l^1,l^2$, it must be the case that $P$ intersects $S\setminus [X]$ in a vertex, call it $a$. However, by assumption, there is a path from $v_\alpha$ to $a$ in the graph $H_I-[X]$. Since the internal vertices of $P$ are disjoint from $[X]$,  we obtain a walk (and hence a connected component) in $H_I-[X]$ that contains $v_\alpha$, $l^1$ and $l^2$, contradicting our premise that there is a feasible labeling of $G-X$ which labels $v$ with $\alpha$.
   
   In the second case, since $l^1\in Z$ and $l^2\notin Z$, the subpath of $P$ between $l^1$ and $l^2$ must intersect $S\setminus [X]$ at a vertex, call it $a$. Since $P$ is internally disjoint from $[u]$, we conclude that
 the subpath of $P$ between $l^1$ and $l^2$ is disjoint from $[X]$ and hence present in $H_I-[X]$. However, we assumed that $a$ is reachable from $v_\alpha$ in $H_I-[X]$, implying that $v_\alpha$ can reach both $l^1$ and $l^2$ in $H_I-[X]$, contradicting the existence of a feasible labeling of $G-X$ (Observation \ref{obs:color_propagation}) which labels $v$ with $\alpha$. Hence, we conclude that $P$ is regular and this completes the proof of the claim.
\end{proof}

%
%
%
%
%
%

We now return to the proof of the first case.
Since $u_b\in Z$ and $u_{b'}\notin Z$ (as $N[Z]$ is regular), $P$ must intersect $N(Z)$ which is the same as $S$, in $S\setminus [X]$. Furthermore, $P$ must intersect $N(C)$ where $C$ is the connected component of $Z'=[Z]\setminus Z$ containing the vertex $u_{b'}$. We now have the following 2 subcases based on the intersection of $P$ with the(not necessarily non-empty) set $S\cap N(C)$. In both subcases we will demonstrate the presence of a $v_\alpha$-$v_\beta$ path in $H_I-[X]$ for some $\beta\in \Sigma\setminus \{\alpha\}$.

\begin{figure}[t]
  \includegraphics[height= 200 pt, width=400 pt]{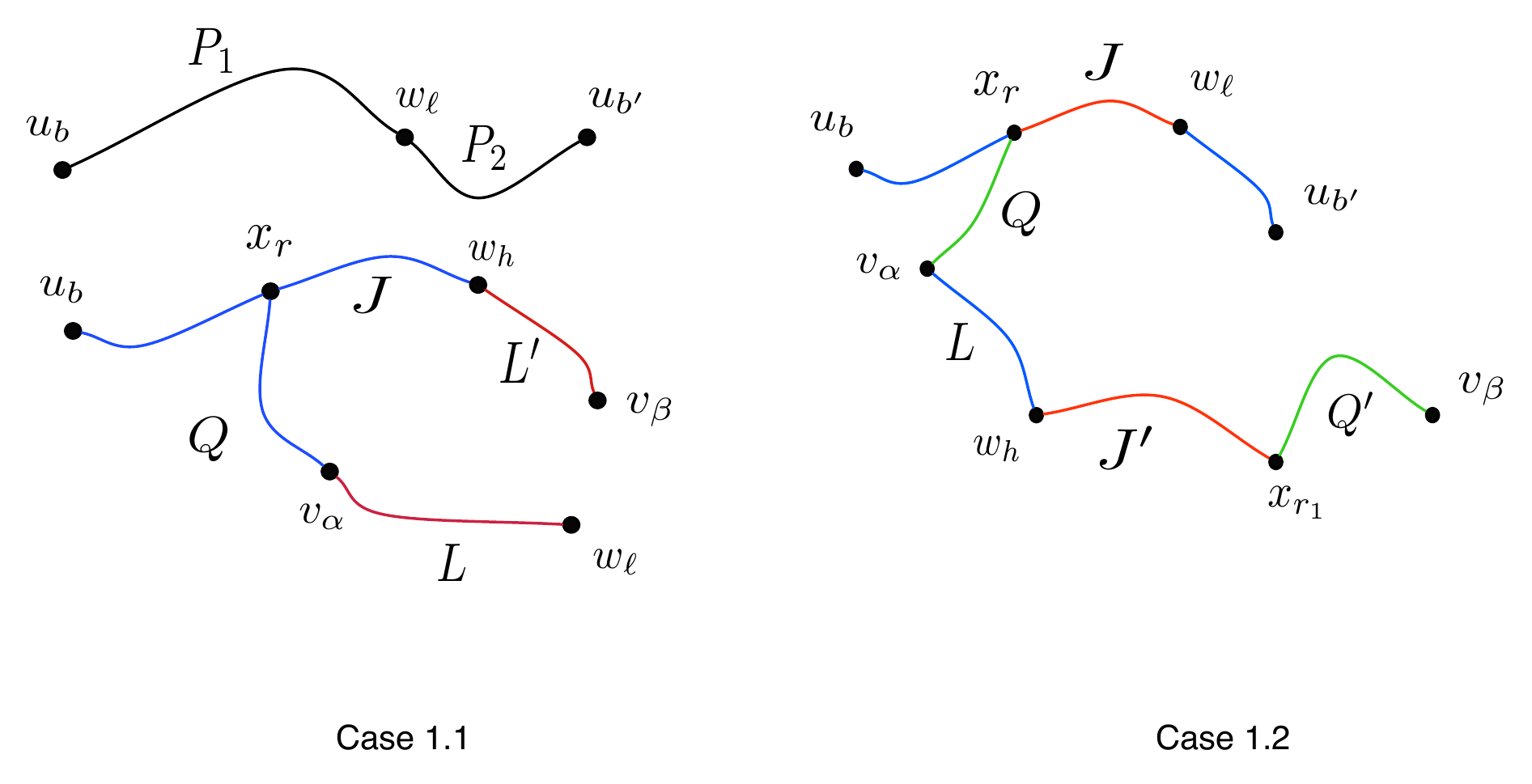}
\caption{Illustrations of the paths used in the arguments in Case 1. Note that contrary to the illustration, some of these paths may intersect.}
  \label{fig:case1}
\end{figure}

\begin{description}
\item[Case 1.1:] \emph{ $P$ contains a vertex in $S\cap N(C)$}. Let $w_\ell$ be a vertex in $S\cap N(C)$ which appears in $P$. We let $P_1$ denote the subpath of $P$ from $u_b$ to $w_\ell$ and $P_2$ denote the subpath of $P$ from $w_\ell$ to $u_{b'}$
 (see Figure \ref{fig:case1}). Furthermore, since $P$ is internally regular, $P_1$ and $P_2$ are regular. We apply Lemma \ref{lem:path_copies} to the regular path $P_2$ to get a path $P_2'$ with $u_b$ as one endpoint and $w_h$ as the other endpoint, where $w_h\neq w_\ell$. Now, since $W_\ell\in N[Z]$ and $N[Z]$ is regular by our assumption, it must be the case that $w_h\notin Z$. Therefore the path $P_2'$ must intersect $S$ at a vertex other than $w_\ell$. Let $x_r$ be such a vertex, where $x\in V(G)$ and $r\in \Sigma$. However, in the case we are in, we know that $x_r$ (which is contained in $S\setminus [X]$) is reachable from $v_\alpha$ in $H_I-[X]$ by a path $Q$ whose internal vertices lie in $Z$. We let the subpath of $P_2'$ from $x_r$ to $w_h$ be denoted by $J$. Furthermore, the case we are in guarantees that $w_\ell$ is reachable from $v_\alpha$ in $H_I- [X]$ via a path $L$ whose internal vertices lie in $Z$. Since $L$ lies completely in $N[Z]$, it is regular and we may apply Lemma \ref{lem:path_copies} on this path to obtain a path $L'$ with $w_h$ s one endpoint and $v_\beta$ as the other endpoint for some $\beta\in \Sigma$. Since we have already argued that $w_h\neq w_\ell$, it follows that $\beta\neq \alpha$. Therefore, we get a concatenated walk $Q+J+L'$ which is a  walk that is present in the graph $H_I-[X]$ and contains $v_\alpha$ and $v_\beta$,  contradicting the premise of the lemma that there is a feasible labeling for $G-X$ setting $v$ to $\alpha$. This completes the argument for this subcase.

\item[Case 1.2:] \emph{$P$ does not contain a vertex in $S\cap N(C)$}. Let $x_r$ be the last vertex of $S$ which is encountered when traversing $P$ from $u_b$ to $u_{b'}$ and let $w_\ell$ be the last vertex of $N(C)$ encountered in the same traversal. Observe that since the previous subcase does not hold, it must be the case that $x_r$ occurs before $w_\ell$ in this traversal. We let $J$ denote the subpath of $P$ between $x_r$ and $w_\ell$. Now, Lemma \ref{lem:partial_symmetry} implies that there is a $h\in \Sigma\setminus \{\ell\}$ such that $w_h\in S$. This is because $N(C)\subseteq [S]$. 
Now, the case we are in guarantees the presence of paths $L$ and $Q$ from $v_\alpha$ to $w_h$ and $x_r$ respectively such that $L$ and $Q$ both lie strictly inside $N[Z]$ and hence are regular. Now, we apply Lemma \ref{lem:path_copies} on the regular path $J$ to get a path $J'$ with $w_h$ as one endpoint and $x_{r_1}$ as the other for some $r_1\in \Sigma$. Since we have already argued that $w_h\neq w_\ell$, it must be the case that $r_1\neq r$. Now, we apply Lemma \ref{lem:path_copies} on the regular path $Q$ to get a path $Q'$ with $x_{r_1}$ as one endpoint and $v_\beta$ as the other for some $\beta\in \Sigma$. Since we have shown that $r_1\neq r$, we infer that $\beta\neq \alpha$. Now, the concatenated walk $L+J'+Q'$ implies the presence of a $v_\alpha$-$v_\beta$ path in $H_I-[X]$, a contradiction to the premise of the lemma. This completes the argument for this subcase.

\end{description}

Thus we have concluded that $G-X'$ has a feasible labeling, contradicting the minimality of $X$. This completes the argument for the first case and we now move on to the second case.

\medskip

\item[Case 2: $S'$ is non-empty.]

Let $\cal Q$ be a set of $\vert S\vert$-many $v_\alpha$-$S$ paths contained entirely in $N[Z]$ which are vertex disjoint except for the vertex $v_\alpha$. Since $S$ is a minimum $v_\alpha$-$T$ separator, such a set of paths exists. Recall that $X_1$ denotes the set $X\cap Z\inv$. We let $\hat X_1$ denote the set $[X]\cap Z$. That is, those copies of $X_1$ present in $Z$. Due to the presence of the set of paths $\cal Q$ and the fact that $v$ is disjoint from $X$, it must be the case that $\hat X_1$ contains at least one vertex in each path in $\cal Q$ that connects $v$ and $S'$. Furthermore, since $S$ is a good separator, we conclude that $|X_1|=|(\hat X_1)\inv|\geq |(S')\inv|$. We now claim that $X'=(X\setminus X_1)\cup (S')\inv$ is also a solution for the given instance. That is, $|X'|\leq |X|$ and $G-X'$ has a feasible labeling. By definition, $|X'|\leq |X|$ holds. Therefore, it remains to prove that $G-X'$ has a feasible labeling.

    Again, it must be the case that any connected component of $G-X'$ which does not have a feasible labeling must intersect the set $X_1$. Any other component of $G-X'$ is contained in a component of $G-X$ and already has a feasible labeling by the premise of the lemma.

        By Lemma \ref{lem:characterization}, there must be a vertex $u^1\in X_1$ and distinct labels $b,b'\in \Sigma$ such that $u^1_b\in Z$ and there is a $u^1_b-u^1_{b'}$ path $P$ in $H_I-[X']$. We now consider the intersection of $P$ with the set $[X_1]$ and let $p_{\gamma_1}$ and $q_{\gamma_2}$ be vertices on $P$ such that $p_{\gamma_1},q_{\gamma_2}\in [Z]$, the subpath of $P$ from $p_{\gamma_1}$ to $q_{\gamma_2}$ is internally disjoint from $[X_1]$ \emph{and} $p_{\gamma_1}\in Z$ and $q_{\gamma_2}\notin Z$. We first argue that such a pair of vertices exist.

           We begin by setting $p_{\gamma_1}=u^1_b$ and $q_{\gamma_2}=u^1_{b'}$. If the path $P$ is already internally disjoint from $[X_1]$ then we are done. Otherwise, let $u^2_c$ be the vertex of $[X_1]$ closest to $p_{\gamma_1}$ along the subpath between $p_{\gamma_1}$ and $q_{\gamma_2}$. Now, if $u^2_c$ is not in $Z$ then we are done by setting $q_{\gamma_2}=u^2_c$. Otherwise, we continue by setting $p_{\gamma_1}=u^2_c$. Since this process must terminate, we conclude that the vertices $p_{\gamma_1}$ and $q_{\gamma_2}$ with the requisite properties must exist.

        For ease of notation we will now refer to the path between $p_{\gamma_1}$ and $q_{\gamma_2}$ as $P$. Note that by definition, $P$ is internally disjoint from $[X_1]$. We now have a claim identical to that in the previous case.
        
        \medskip
        
        \begin{claim}
        The path $P$ is internally regular.	
        \end{claim}

\begin{proof}
 The proof of this claim is identical to the previous case and only uses the fact that one endpoint of $P$ is inside $Z$, the other outside $N[Z]$ and that $P$ is internally disjoint from $[X_1]$. Using these properties, one can argue that if $P$ is not internally regular, then $v_\alpha$ is in the same component as a pair of vertices in $[l]$ in the graph $H_I-[X]$, for some $l\in V(G)$, contradicting the premise of the lemma.
 \end{proof}
 
 We now complete the proof of this case. Since $p_{\gamma_1}\in Z$ and $q_{\gamma_2}\in [Z]\setminus Z$, $P$ must intersect $N(Z)$ in $(S\setminus [X])\setminus S'$. Furthermore, $P$ must also intersect $N(C)$ where $C$ is the connected component of $H_I[Z']$ containing $q_{\gamma_2}$, where $Z'=[Z]\setminus Z$. We again consider 2 subcases based on the intersection of the path $P$ with the (not necessarily non-empty) set $N(C)\cap S$.
 
 \begin{figure}[t]
  \includegraphics[height= 200 pt, width=420 pt]{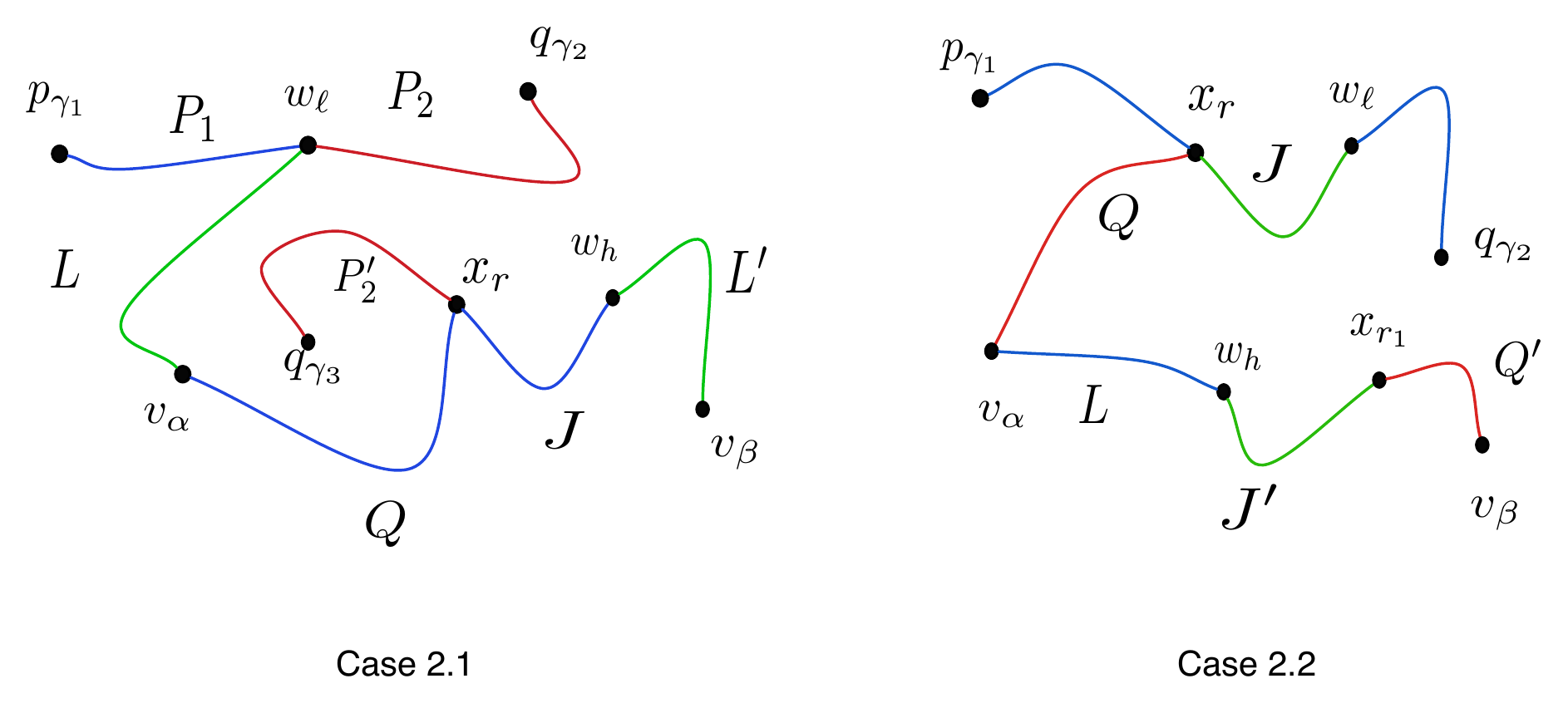}
  \caption{Illustrations of the paths used in the arguments in Case 2. Note that contrary to the illustration, some of these paths may intersect.}
  \label{fig:case2}
\end{figure}

 \begin{description}
 \item[Case 2.1:] \emph{ $P$ contains a vertex in $S\cap N(C)$.} Let $w_\ell$ be a vertex in $S\cap N(C)$ which appears in $P$. We let $P_1$ denote the subpath of $P$ from $p_{\gamma_1}$ to $w_\ell$ and $P_2$ denote the subpath of $P$ from $w_\ell$ to $q_{\gamma_2}$ (see Figure \ref{fig:case2}). Since $P$ is internally regular, $P_1$ and $P_2$ are regular. Furthermore, since $q_{\gamma_2}\notin Z$, there is a $\gamma_3\in \Sigma\setminus \{\gamma_2\}$ such that $q_{\gamma_3}\in Z$. We now apply Lemma \ref{lem:path_copies} on the regular path $P_2$ to get a path $P_2'$ with $q_{\gamma_3}$ as one endpoint and $w_h$ as the other, where $h\neq \ell$ since $\gamma_2\neq \gamma_3$. Furthermore, since $w_\ell\in N[Z]$ and $N[Z]$ is regular, it must be the case that $w_h\notin Z$. Therefore the path $P_2'$ must intersect $N(Z)$ at a vertex $x_r$. Let $J$ be the subpath of $P_2'$ from $x_r$ to $w_h$. Now, since $x_r\in (S\setminus [X])\setminus S'$, we know that there is a $v_\alpha$-$x_r$ path in $H_I-[X]$ which lies entirely in $N[Z]$. Let $Q$ be such a path. Similarly,  we know that there is a $v_\alpha$-$w_\ell$ path $L$ in $H_I-[X]$ which also lies entirely in $N[Z]$ and hence is regular. We now apply Lemma \ref{lem:path_copies} on $L$ to get a path $L'$ with $w_h$ as one endpoint and $v_\beta$ as the other endpoint for some $\beta\in \Sigma$. Since we have already argued that $w_h\neq w_\ell$, we conclude that $\beta\neq \alpha$. However, the concatenated walk $Q+J+L'$ is present in $H_I-[X]$, implying a $v_\alpha$-$v_\beta$ path in $H_I-[X]$, a contradiction to the premise of the lemma. We now address the second subcase under the assumption that this subcase does not occur.

 \item[Case 2.2:] \emph{ $P$ does not contain a vertex in $S\cap N(C)$.} Let $x_r$ be the last vertex of $S$ which is encountered when traversing $P$ from $p_{\gamma_1}$ to $q_{\gamma_2}$ and let $w_\ell$ be the last vertex of $N(C)$ encountered in the same traversal. Since the previous subcase is assumed to not hold, $x_r$ must occur before $w_\ell$ in this traversal. We let $J$ denote the subpath of $P$ between $x_r$ and $w_\ell$. Lemma \ref{lem:partial_symmetry} implies the existence of a  label $h\in \Sigma\setminus \{\ell\}$ such that $w_h\in S$. This follows from the fact that $N(C)\subseteq [S]$. Also, since $w_\ell$ occurs in $P$, $w_h$ is not contained in $S'$ or $[X]$. The same holds for $x_r$ Therefore, the case we are in guarantees the presence of paths $L$ and $Q$ from $v_\alpha$ to $w_h$ and $x_r$ respectively, where $L$ and $Q$ are contained within the set $N[Z]$ and hence they must be regular and amenable to applications of Lemma \ref{lem:path_copies}. We begin by applying Lemma \ref{lem:path_copies} on the regular path $J$ to get a path $J'$ with $w_h$ as one endpoint and $x_{r_1}$ as the other for some $r_1\in \Sigma$. However, since $h\neq \ell$, we conclude that $r_1\neq r$. Therefore, we now apply Lemma \ref{lem:path_copies} on the path $Q$ to obtain a path $Q'$ with $x_{r_1}$ as one endpoint with the other endpoint being $v_\beta$ for some $\beta\in \Sigma$. Again, since $r_1\neq r$, we conclude that $\beta\neq \alpha$. Now, observe that the concatenated walk $L+J'+Q'$ implies the presence of a $v_\alpha$-$v_\beta$ path in $H_I-[X]$, a contradiction to the premise of the lemma. This completes the argument for this subcase as well and consequentially that for Case 2.
\end{description}

	\end{description}
	
	We have thus proved that Case 1 cannot occur at all and in Case 2, there is an exchange argument which constructs an alternate solution $X'$ which is  disjoint from $Z$. This completes the proof of the lemma.	
\end{proof}

The main consequence of the above lemma is that at any point in the run of our algorithm solving an instance $I=(G,k,\phi,\tau)$, if there is a vertex $v$ whose label is `fixed', i.e. $\tau(v)=\{\alpha\}$ for some $\alpha\in \Sigma$ and there is a good $v_\alpha$-$T$ separator $S$ where $T$ is defined as in the premise of the above lemma, then we can correctly `fix' the labelings of all vertices in the set $(R(v_\alpha,S))\inv$. That is, we can define a new function $\tau'$ as follows. For every $u\in V(G)$ and $\gamma\in \Sigma$, we set $\tau'(u)=\{\gamma\}$ if $u_\gamma\in R(v,\alpha)$ and $\tau'(u)=\tau(u)$ otherwise. Lemma \ref{lem:reduction} implies that the given graph has a  deletion set of size at most $k$ which leaves a graph with a feasible labeling consistent with $\tau$ if and only if the graph has deletion set of size at most $k$ which leaves a graph with a feasible labeling consistent with $\tau'$. 

\subsection{Computing good separators}

\begin{definition}
	Let $G$ be a graph and $X$ and $Y$ be disjoint vertex sets and $S$ a \emph{minimum} $X$-$Y$ separator. We say that $S$ is a minimum  $X$-$Y$ separator \textbf{closest} to $X$ if there is no $S'$ which is a minimum $X$-$Y$ separator such that $R(X,S')\subset R(X,S)$. We say that $S$ is a minimum  $X$-$Y$ separator \textbf{closest} to $Y$ if there is no $S'$ which is a minimum $X$-$Y$ separator such that $R(X,S')\supset R(X,S)$. We let $\lambda(X,Y)$ denote the size of a minimum $X$-$Y$ separator.

\end{definition}

\begin{lem}[\cite{Marx06}]
Let $G$ be a graph and $X$ and $Y$ be disjoint vertex sets. There is a unique minimum $X$-$Y$ separator closest to $X$ and a unique minimum $X$-$Y$ separator closest to $Y$.
\end{lem}

We need the following lemma from \cite{MarxOR13}.

\begin{lem}[\cite{MarxOR13}]\label{lem:separator layer}
 Let $X,Y$ be two disjoint vertex sets in a graph $G$ such that the minimum size of an $X$-$Y$ separator is $\ell>0$. Then, there is a collection ${\cal J}=\{J_1,\dots,J_q\}$ of vertex sets where $X\subseteq J_i\subseteq V(G)\setminus Y$ such that
 
 \begin{enumerate}
\setlength{\itemsep}{-2pt}
  \item $J_1\subset J_2\subset \cdots \subset J_q$,
  \item $J_i$ is reachable from $X$ in $G[J_i]$,
  \item $\vert N(J_i)\vert=\ell$ for every $1\leq i\leq q$ and 
  \item every $X$-$Y$ separator of size $\ell$ is fully contained in $\bigcup_{i=1}^q N(J_i)$.
 \end{enumerate}
 
 \noindent
 Furthermore, there is an algorithm that, given $G,X,Y$ and an integer $\ell'$, runs in time $\Oh(\ell'(\vert V(G)\vert +\vert E(G)\vert))$ and either correctly concludes that there is no $X$-$Y$ separator of size at most $\ell'$ or produces the sets $J_1,J_2\setminus J_1,\dots, J_q\setminus J_{q-1}$ corresponding to the aforementioned  collection $\cal J$.

\end{lem}

We will state some simple consequences of the above lemma in a form that will be easier to invoke during our arguments.

\begin{lem}\label{lem:easy_invocation}
Let $X,Y$ be two disjoint vertex sets in a graph $G$ such that the minimum size of an $X$-$Y$ separator is $\ell>0$. Let ${\cal J}=\{J_1,\dots,J_q\}$ be the collection defined in  the statement of Lemma \ref{lem:separator layer}.  Then, 
\begin{enumerate}

  \item $N(J_1)$ is the minimum $X$-$Y$ separator closest to $X$ and $N(J_q)$ is the minimum $X$-$Y$ separator closest to $Y$. 
	\item  $\forall v\in J_1$, the size of a smallest minimal $X$-$Y$ separator which contains $v$ is at least $\ell+1$.
	\item  for any $1\leq i\leq q-1$, for any vertex $v\in J_{i+1}\setminus N[J_i]$, the size of any minimal $X$-$Y$ separator which contains $v$ is at least $\ell+1$.   
\end{enumerate}
	
\end{lem}

\begin{proof} For the first statement, if $S=N(J_1)$ is not the minimum $X$-$Y$ separator closest to $X$, then there is a minimum $X$-$Y$ separator $S'$ such that $R(X,S')\subset R(X,S)$. Since $R(X,S)=J_1$, it must be the case that $S'\cap J_1\neq \emptyset$. However, by Lemma \ref{lem:separator layer}, we know that every minimum $X$-$Y$ separator is contained in $\bigcup_{i=1}^q N(J_i)$ and since $J_1\subset J_2\subset \dots J_q$, every minimum $X$-$Y$ separator is disjoint from $J_1$, a contradiction. Similarly, we can argue that $N(J_q)$ is the unique minimum $X$-$Y$ separator closest to $Y$. 

The argument for the second statement is identical. For the last statement, we argue that $\bigcup_{j=1}^q N(J_j)$ is disjoint from $J_{i+1}\setminus N[J_i]$ for any $1\leq i\leq q-1$. Suppose that this is not the case and there is an index $1\leq i\leq q-1$ and a vertex $u$ such that $u$ is present in both sets $\bigcup_{j=1}^q N(J_j)$ and $J_{i+1}\setminus N[J_i]$. Since $u\notin N[J_i]$, $u\notin N[J_r]$ for any $r\leq i$. Further, since $J_{r}\supseteq J_{i+1}$ for every $r>i$, $u$ is present in $J_r$ for every $r>i$ and hence not in $N(J_r)$ for any $r>i$. This  contradicts our assumption and completes the proof of the lemma.	
\end{proof}

\begin{lem}\label{lem:crux lemma}
Let $I=(G,k,\phi,\tau)$ be an instance of {\labelcoverv}, $v$ be a vertex in $G$ and let $\alpha\in \Sigma$. Let $T^\alpha_v$ denote the set $[v]\setminus \{v_\alpha\}$ and $T\supseteq T^\alpha_v$ be a set not containing $v_\alpha$. There is an algorithm that, given $I$, $v$, $\alpha$, and $T$ runs in time $\Oh(|\Sigma| \cdot k(m+n))$ and either 
	\begin{itemize} 
	\item correctly concludes that there is no $v_\alpha$-$T$ separator of size at most $|\Sigma|\cdot k$ or 

   \item returns a pair of minimum $v_\alpha$-$T$ separators $S_1$ and $S_2$ such that $S_2$ covers $S_1$, $S_1$ is good, $S_2$ is bad and for any vertex $u\in R(v_\alpha,S_2)\setminus R[v_\alpha,S_1]$, the size of a minimal $v_\alpha$-$T$ separator containing $u$ is at least $|S_1|+1$  or
   \item returns a good minimum $v_\alpha$-$T$ separator $S$ such that no other minimum  $v_\alpha$-$T$ separator covers $S$ or
   \item correctly concludes that there is no good $v_\alpha$-$T$ minimum separator.
		
	\end{itemize}
 \end{lem}

\begin{proof}
	We begin by executing the algorithm of Lemma \ref{lem:separator layer} with $G=H_I$, $X=\{v_\alpha\}$, $Y=T$ and $\ell'=|\Sigma|\cdot k$.
	  If this algorithm concluded that there is no $v_\alpha$-$T$ separator of size at most $|\Sigma|\cdot k$, then we return the same. Otherwise, it must be the case that this subroutine returned a family of sets ${\cal J}=\{J_1,\dots,J_q\}$ with $|N(J_i)|\leq |\Sigma|\cdot k$ for each $i\in [q]$. 
	
	We examine the sets in $\cal J$ in time $\Oh(|\Sigma| \cdot (m+n))$ and compute the least $i$ such that $N[J_i]$ is irregular. Suppose $i=1$. Since every minimum $v_\alpha$-$T$ separator covers the minimum $v_\alpha$-$T$ separator closest to $v_\alpha$ (which is precisely $N(J_1)$ by Lemma \ref{lem:easy_invocation}), there is no good minimum $v_\alpha$-$T$ separator at all.
	
	On the other hand, if $i>1$, then we set
	 $S_1=N(J_{i-1})$ and $S_2=N(J_i)$. It follows from Lemma \ref{lem:easy_invocation} that these sets satisfy the required properties.
	
	Finally, if $N[J_i]$ is regular for every $i\in [q]$, then we set $S=N(J_q)$. It already follows from Lemma \ref{lem:easy_invocation} that no other minimum $v_\alpha$-$T$ separator covers $N(J_q)$.
	This completes the proof of the lemma.
	\end{proof}

\begin{lem}\label{lem:twopaths}
Let $I=(G,k,\phi,\tau)$ be an instance of {\labelcoverv}, $v$ be a vertex in $G$, $\alpha\in \Sigma$, $T\supseteq [v]\setminus \{v_\alpha\}$ be a set not containing $v_\alpha$ and let $\ell>0$ be the size of a minimum $v_\alpha$-$T$ separator in $H_I$. Let $S_1$ and $S_2$ be a pair of minimum $v_\alpha$-$T$ separators such that 
  $S_1$ is good,
  $S_2$ is bad, and
 and for any vertex $y\in R(v_\alpha,S_2)\setminus R[v_\alpha,S_1]$, the size of a minimal $v_\alpha$-$T$ separator containing $y$ is at least $\ell+1$. Let $u\in V(G)$ and $\gamma_1,\gamma_2\in \Sigma$ such that $u_{\gamma_1},u_{\gamma_2}\in R[v_\alpha,S_2]$.
Then, 
\begin{enumerate}

 \item $R[v_\alpha,S_2]$ contains a pair of paths $P_1$ and $P_2$ such that for each $i\in \{1,2\}$, the path $P_i$ is a $v_\alpha$-$u_{\gamma_i}$ path and both paths are internally vertex disjoint from $S_2$ and contain at most one vertex of $S_1$. 
 \item Given $I$, $v_\alpha$, $S_1$ and $S_2$, there is an   algorithm that, in time $\Oh(|\Sigma| \cdot k(m+n))$,  computes a pair of paths with the above properties.
 \item For $i\in \{1,2\}$, any minimum $v_\alpha$-$T\cup \{u_{\gamma_i}\}$
separator disjoint from $V(P_i)\cap (S_1\cup S_2)$ and $R(v_\alpha,S_1)$ has size at least $\ell+1$, where $\ell$ is the size of a minimum $v_\alpha$-$T$ separator.
 
 \end{enumerate}

 \end{lem}

\begin{proof} We begin with the first statement. Since for each $i\in \{1,2\}$, $u_{\gamma_i}$ is in $R[v_\alpha,S_2]$, there is clearly a $v_\alpha$-$u_{\gamma_i}$ path which is internally vertex disjoint from $S_2$. Let $Q_i$ be a such a path containing minimum possible vertices of $S_1$. If $Q_i$ intersects $S_1$ at most once then we are done. Otherwise $Q_i$ intersects $S_1$ at least twice in vertices $x^i_{r_1}$ and $y^i_{r_2}$ with $x^i_{r_1}$ the vertex closer to $v_\alpha$. Since $S_1$ is a minimum $v_\alpha$-$T$ separator, we have the existence of a $v_\alpha$-$y^i_{r_2}$ path which is internally disjoint from $S_1$. Therefore, we can use this path and the subpath of $Q_i$ from $y^i_{r_2}$ to $u_{\gamma_i}$ to obtain a walk from $v_\alpha$ to $u_{\gamma_i}$, which  contains fewer vertices of $S_1$ than $Q_i$, a contradiction to the choice of $Q_i$. This completes the proof of the first statement.

For the second statement, observe that if $u_{\gamma_i}\in R[v_\alpha,S_1]$ then such a path can be found by a simple BFS from $v_\alpha$. On the other hand, if $u_{\gamma_i}\in R[v_\alpha,S_2]\setminus     R[v_\alpha,S_1]$ then we can find the path $P_i$ by computing an arbitrary path from $u_{\gamma_i}$ to a vertex $x_r\in S_1$ such that this path is internally vertex disjoint from $S_1$ and $S_2$. We then compute a path from $v_\alpha$ to $x_r$ which is internally vertex disjoint from $S_1$. Since this can be achieved by 2 applications of a standard Breadth First Search, the claimed bound on the running time holds.

For the final statement, observe that any $v_\alpha$-$T\cup \{u_{\gamma_i}\}$ separator	 disjoint from $V(P_i)\cap (S_1\cup S_2)$ must contain a vertex in  $R(v_\alpha,S_1)$ or  $R(v_\alpha,S_2)\setminus R[v_\alpha,S_1]$. By the premise of the lemma, any such separator must have size at least $\ell+1$. This completes the proof of the lemma.    			
\end{proof}

We are now ready to prove Theorem \ref{thm:main_theorem} by describing our algorithm for {\labelcoverv}. Before doing so, we  make the following important remark regarding the way we use the algorithms described in this subsection. In the description of our main algorithm, there will be points where we make a choice to \emph{not} delete certain vertices. That is, we will choose to exclude them from the solution being computed. At such points, we say that we make these vertices \emph{undeletable}. 

All the above algorithms also work when given an undeletable set of vertices in the graph and the minimum separators we are looking for are the minimum \emph{among} those separators disjoint from the undeletable set of vertices. Regarding the running time of these algorithms, there will be a multiplicative factor of $|\Sigma|\cdot k$ which arises due to potentially blowing up the size of the graph by a factor of $|\Sigma| \cdot k$ by making $(|\Sigma|\cdot k)+1$ copies of every undeletable vertex. \\

\section{The Linear time algorithm for {\labelcoverv}}

%
%
%
%
%
%

\subsection{Description of the algorithm.}
Before we describe our algorithm, we state certain assumptions we make regarding the input. We assume that at any point, we are dealing with a connected graph $G$. 
 Furthermore, we assume that instances of {\labelcoverv} are  given in the form of a tuple -- $(G,k,\phi,\tau,w^*,V^\infty)$ where the element $w^*$ denotes either a vertex from $V(G)$ or it is undefined. If $w^*$ denotes a vertex then, $|\tau(w^*)|=1$ and we will attempt to solve the problem on the tuple $(G,k,\phi,\tau,w^*,V^\infty)$ under the assumption that $w^*$ is not in the solution (which is required to be disjoint from $V^\infty$). Furthermore the definition of the problem allows us to assume that if there is a feasible labeling for this instance (after deleting a solution) then there is one consistent with $\tau$. Since $\tau(w^*)$ is singleton, any feasible labeling consistent with $\tau$ must set $w^*$ to the unique label in $\tau(w^*)$.


We first check if $G$ already has a feasible labeling (not necessarily one consistent with $\tau$). If so, then we are done. If not and $k=0$ then we return {\No}. If any connected component of $G$ has a feasible labeling then we remove this component. Otherwise, we check if $w^*$ is defined. If $w^*$ is undefined, then we pick an arbitrary \emph{deletable} vertex $v\in V(G)$. That is $v\notin V^\infty$. We then recursively solve the problem on the instances $I_{q_0},\dots, I_{q_r}$ where $\{q_1,\dots, q_r\}=\tau(v)$ and for each $q_i$ where $i\geq 1$, the instance 
$I_{q_i}$ is defined to be $(G,k,\phi,\tau_{v=q_i},w^*,V_1^\infty)$ with $\tau_{v=q_i}$ defined as the function obtained from $\tau$ by restricting the image of $v$ to the singleton set $\{q_i\}$, $w^*$ defined as $w^*=v$ and $V_1^\infty$ defined as $V_1^\infty=V^\infty \cup\{v\}$. The instance $I_{q_0}$ is defined as $(G-\{v\},k-1,\phi',\tau',w^*,V^\infty)$ where $\phi'$ and $\tau'$ are restrictions of $\phi$ and $\tau$ to the graph $G-\{v\}$. This will be the only branching rule which has a branching factor depending on the parameter (in this case the size of the label set $\Sigma$) and we call this rule, $\mathbf{B_0}$.

   We now describe the steps executed by the algorithm in the case when $w^*$ is defined. Suppose that $w^*=v$, $\tau(v)=\alpha$. Recall that by our assumption regarding well-formed inputs, if $w^*$ is defined then $\tau(w^*)$ must be a singleton set. We set 
   $T=\bigcup_{u\in V(G)} \bigcup_{\gamma\in \Sigma\setminus \tau(u)} u_\gamma $. Intuitively, $T$ is the set of all vertices $u_\gamma$ such that if there is a feasible labeling of $G$ (after deleting the solution) which sets $v$ to $\alpha$ then it cannot be consistent with $\tau$ unless the solution hits all paths in $H_I$ (where $I$ is the given instance) between $v_\alpha$ and $u_\gamma$.  We remark that since $T$ depends only on the input instance $I$, we use $T(I)$ to denote the set $T$ corresponding to any input instance $I$. Once we set $T$ as described we first check if there is a $v_\alpha$-$T$ path in $H_I$. If not, then the algorithm deletes the component of $G$ containing $v$ and recurses by setting $w^*$ to be undefined. The correctness of this operation is argued as follows. Observe that $T$ contains all vertices of $[v]\setminus \{v_\alpha\}$ and excludes $v_\alpha$. Therefore, Lemma \ref{lem:characterization} implies that the component of $G$ containing $v$ already has a feasible labeling and hence can be removed.  
   
   Otherwise if there is a $v_\alpha$-$T$ path in $H_I$, then we execute the algorithm of Lemma \ref{lem:crux lemma} with this definition of $v$, $\alpha$ and $T$ and undeletable set $[V^\infty]$. Observe that $T$ contains all vertices of $[v]\setminus \{v_\alpha\}$ but excludes $v_\alpha$. This is because $\tau(v)=\{\alpha\}$. The next steps of our algorithm depend on the output of this subroutine. For each of the four possible  outputs, we describe an exhaustive branching.

    \begin{description}
   
    	\item[Case 1:] \emph{ The subroutine returns that there is no $v_\alpha$-$T$ separator of size at most $|\Sigma| \cdot k$ which is disjoint from $[V^\infty]$}. In this case, our algorithm returns {\sc No}. The correctness of this step follows from Lemma \ref{lem:reduction}.
    	\medskip

    	\item[Case 2:] \emph{The subroutine returns a good $v_\alpha$-$T$ separator $S$ which is smallest among all $v_\alpha$-$T$ separators disjoint from $[V^\infty]$ such that no other $v_\alpha$-$T$ separator disjoint from $[V^\infty]$ and having the same size as $S$, covers $S$.} In this case, we do the following. For each vertex $u_\gamma$ in the set $R(v_\alpha,S)$ where $u\in V(G)$ and $\gamma\in \Sigma$, we set $\tau(u)=\{\gamma\}$ and add $u$ to $V^\infty$. That is, we set $V^\infty=V^\infty\cup (R(v_\alpha,S))\inv$. Note that prior to this operation, $\gamma\in \tau(u)$ since otherwise $u_\gamma$ would belong to $T$. We then pick an arbitrary vertex $x_\delta\in S$ and recursively solve the problem on 2 instances $I_1$ and $I_2$ defined as follows. The instance $I_1$ is defined to be $(G-\{x\},k-1,\phi',\tau',V^\infty)$ where $\phi'$ and $\tau'$ are restrictions of $\phi$ and $\tau$ to $G-\{x\}$. The instance $I_2$ is defined to be $(G,k,\phi,\tau',V_1^\infty)$ where $V_1^\infty=V^\infty\cup \{x\}$ and $\tau'$ is defined to be the same as $\tau$ on all vertices but $x$ and $\tau'(x)=\{\delta\}$. We call this branching rule, $\mathbf{B_1}$. The exhaustiveness of this branching step follows from the fact that once the vertices in $(R(v_\alpha,S)\inv)$ are made undeletable, unless the vertex $x$ is deleted, Observation \ref{obs:color_propagation} forces any feasible labeling that labels $v$ with $\alpha$ to label $x$ with $\delta$.

    	\medskip

    	\item[Case 3:] \emph{The subroutine correctly concludes that there is no good $v_\alpha$-$T$ separator which is also smallest among all $v_\alpha$-$T$ separators disjoint from $[V^\infty]$.} In this case, we compute $S$, the minimum $v_\alpha$-$T$ separator that is disjoint from $V^\infty$ and closest to $v_\alpha$. Since $S$ is not good, $R[v_\alpha,S]$ contains a pair of vertices $u_{\gamma_1}$ and $u_{\gamma_2}$ for some $u\in V(G)$ and $\gamma_1,\gamma_2\in \Sigma$. Furthermore, since $S$ is a $v_\alpha$-$T$ separator, it must be the case that $u_{\gamma_1}$ and $u_{\gamma_2}$ are not in $T$. This implies that $\{\gamma_1,\gamma_2\}\subseteq \tau(u)$. We now recursively solve the problem on 3 instances $I_0$, $I_1$, $I_2$ defined as follows. The instance $I_0$ is defined as $(G-\{u\},k-1,\phi',\tau',w^*,V^\infty)$, where $\phi'$ and $\tau'$ are defined as the restrictions of $\phi$ and $\tau$ to the graph $G-\{u\}$. The instance $I_1$ is defined as $(G,k,\phi,\tau',w^*,V_1^\infty)$ where $V_1^\infty=V^\infty \cup \{u\}$ and $\tau'$ is defined to be the same as $\tau$ on all vertices but $u$ and $\tau'(u)=\tau(u)\setminus \{\gamma_1\}$. Similarly, the instance $I_2$ is defined as $(G,k,\phi,\tau',w^*,V_1^\infty)$ where $V_1^\infty=V^\infty \cup \{u\}$ and $\tau'$ is defined to be  the same as $\tau$ on all vertices but $u$ and $\tau'(u)=\tau(u)\setminus \{\gamma_2\}$. We call this branching rule $\mathbf{B_2}$.
    	
    	The exhaustiveness of this branching follows from the fact that if $u$ is not deleted (the first branch) then any feasible labeling of $G-X$ for a hypothetical solution $X$ must label $u$ with at most one label out of $\gamma_1$ and $\gamma_2$. Therefore, if $I$ is a {\Yes} instance then for at least one of the 2 instances $I_1$ or $I_2$, there is a feasible labeling of $G-X$ consistent with the corresponding $\tau'$.

%
%
%

    	
    	\item[Case 4:] Finally, we address the case when the subroutine returns a pair of minimum (among those disjoint from $[V^\infty]$) $v_\alpha$-$T$ separators $S_1$ and $S_2$ such that $S_2$ covers $S_1$, $S_1$ is good, $S_2$ is bad and there is no minimum (among those disjoint from $V^\infty$) $v_\alpha$-$T$ separator which covers $S_1$ and is covered by $S_2$. In this case, $R[v_\alpha,S_2]$ contains a pair of vertices $u_\gamma,u_\delta$ for some vertex $u\in V(G)$ and $\gamma,\delta\in \Sigma$.

    	We execute the algorithm of Lemma \ref{lem:twopaths} to compute in time $\Oh(|\Sigma|\cdot k (m+n))$,  a $v_\alpha$-$u_\gamma$ path $P_1$ and a $v_\alpha$- $u_\delta$ path $P_2$ such that both paths are internally vertex disjoint from $S_2$ and contain at most one vertex of $S_1$ each. Let  $x^1,x^2\in V(G)$ and $\beta_1,\beta_2\in \Sigma$ be such that $x^1_{\beta_1}$ and $x^2_{\beta_2}$ are the vertices of $S_1$ in $P_1$ and $P_2$ respectively. Note that $P_1$ or $P_2$ may be disjoint from $S_1$. If $P_i$ ($i\in \{1,2\}$) is disjoint from $S_1$ then we let $x^i_{\beta_i}$ be undefined. We now recurse on the following (at most) 5 
    	instances $I_1,\dots, I_5$ defined as follows.
    	
    	\begin{itemize}
    		
    		\item $I_1=(G-x^1,k-1,\phi',\tau',w^*,V^\infty)$ where $\phi'$ and $\tau'$ are restrictions of $\phi$ and $\tau$ to $G-\{x^1\}$.
    		\item $I_2=(G-x^2,k-1,\phi',\tau',w^*,V^\infty)$ where $\phi'$ and $\tau'$ are restrictions of $\phi$ and $\tau$ to $G-\{x^2\}$.
    		\item $I_3=(G-u,k-1,\phi',\tau',w^*,V^\infty)$ where $\phi'$ and $\tau'$ are restrictions of $\phi$ and $\tau$ to $G-\{u\}$.
    	 
    		\item $I_4=(G,k,\phi,\tau',w^*,V_1^\infty)$ where $V_1^\infty=V^\infty\cup (R(v_\alpha,S_1))\inv\cup \{x^1\}$ and $\tau'$ is the same as $\tau$ on all vertices of $G$ except $u$ and $\tau'(u)=\tau(u)\setminus \{\gamma\}$.

    		\item $I_5=(G,k,\phi,\tau',w^*,V_1^\infty)$ where $V_1^\infty=V^\infty\cup (R(v_\alpha,S_1))\inv\cup \{x^2\}$ and $\tau'$ is the same as $\tau$ on all vertices of $G$ except $u$ and $\tau'(u)=\tau(u)\setminus \{\delta\}$.

    	\end{itemize}
    	
    	This branching rule is called $\mathbf{B_3}$.
    	We argue the exhaustiveness of the branching as follows. The first three branches cover the case when the solution intersects the set $\{x^1,x^2,u\}$. Suppose that a hypothetical solution, say $X$, is disjoint from $\{x^1,x^2,u\}$. By Lemma \ref{lem:reduction}, we may assume that $X$ is disjoint from $R(v_\alpha,S_1)$. Since any feasible labeling of $G-X$ sets $u$ to at most one of $\{\gamma_1,\gamma_2\}$, branching into 2 cases by excluding $\gamma_1$ from $\tau(u)$ in the first case and excluding $\gamma_2$ from $\tau(u)$ in the second case gives us an exhaustive branching.\medskip

    	\end{description}

    	This completes the description of the algorithm. The correctness follows from the exhaustiveness of the branchings.  We will now prove the running time bound stated in the theorem.
    	
    	\medskip
    	
    	\noindent
    	\textbf{Analysis of running time.} It follows from the description of the algorithm and the bounds already proved on the running time of each subroutine, that each step can be performed in time $\Oh((\Sigma+k)^{\Oh(1)}(m+n))$. Therefore, we only focus on bounding the number of nodes in the search tree resulting from this branching algorithm. In order to analyse this number, we introduce the following measure for the instance $I=(G,k,\phi,\tau,w^*,V^\infty)$ corresponding to any node of the search tree. We define $\mu(I)=(\Sigma+1)k-\lambda(I)$ 
    	 where  \medskip

%
%
    	
    	$\lambda(I)$ =
$\left\{
	\begin{array}{ll}
		\lambda(w^*,T(I))  & \mbox{if }  w^* \mbox{ is defined} \\
		0 & \mbox{otherwise } 
	\end{array}
\right.$

    	 Note that $\lambda(w^*,T(I))$ denotes the size of the smallest $w^*$-$T(I)$ separator in $H_I$ among those disjoint from $[V^\infty]$. 
%
%
    	Furthermore, observe that $\mu(I)\leq (|\Sigma|+1)\cdot k$ for any instance on which the algorithm can potentially branch.
    	We now argue that this measure strictly decreases in each branch of every branching rule and since the  number of branches in any branching rule is bounded by $max$ $\{|\Sigma|+1,5\}$ (Rules $\mathbf{B_0}$ and $\mathbf{B_3}$), the time bound claimed in the statement of Theorem \ref{thm:main_theorem} follows. \medskip
    	
\noindent
	\textbf{Rule $\mathbf{B_0}$:} Let $I$ be the instance on which this branching rule is executed and let $I'$ be an instance resulting from an application of this rule. Since $\mathbf{B_0}$ is applicable on $I$, it must be the case that $w^*$ is undefined in $I$ and hence $\lambda(I)=0$. If $k$ drops in $I'$, then it follows from the definition of the measure that $\mu(I')<\mu(I)$. On the other hand, suppose that in $I'$, $w^*$ is defined to be $v_\alpha$ for some $\alpha\in \Sigma$. Since the component of $G$ containing $v$ does not have a feasible labeling and in particular no feasible labeling that sets $v$ to $\alpha$, there is at least one path in $H_I$ (and hence in $H_{I'}$) from $v_\alpha$ to $[v]\setminus \{v_\alpha\}$. As a result, there is at least one path in $H_{I'}$ from $w^*$ to $T(I')$, implying that $\lambda(I')>0$, which in turn implies that $\mu(I')<\mu(I)$.\\


\noindent 
\textbf{Rule $\mathbf{B_1}$:} Observe that $\lambda(I_1)\geq \lambda(I)-|\Sigma|$. Since the budget $k$ drops by 1 for $I_1$, it follows that $\mu(I_1)<\mu(I)$. Furthermore, it follows from Lemma \ref{lem:crux lemma} that $\lambda(I_2)>\lambda(I)$, implying that $\mu(I_2)<\mu(I)$.\\

		\noindent
\textbf{Rule $\mathbf{B_2}$:} Since the budget $k$ drops by 1 for $I_0$ and $\lambda(I_0)\geq \lambda(I)-|\Sigma|$, it follows that $\mu(I_0)<\mu(I)$. Furthermore, it follows from Lemma \ref{lem:crux lemma} that $\lambda(I_1),\lambda(I_2)>\lambda(I)$, implying that $\mu(I_1),\mu(I_2)<\mu(I)$.\\

	\noindent		
\textbf{Rule $\mathbf{B_3}$:} The budget $k$ drops by 1 for the instances $I_1,I_2,I_3$ and for each $i\in \{1,2,3\}$, it holds that $\lambda(I_i)\geq \lambda(I)-|\Sigma|$. Hence, $\mu(I_i)<\mu(I)$ for each $i\in \{1,2,3\}$. For the instances $I_4$ and $I_5$, it follows from Lemma \ref{lem:crux lemma} that $\lambda(I_4),\lambda(I_5)>\lambda(I)$, implying that $\mu(I_4),\mu(I_5)<\mu(I)$.

\medskip

\section{Conclusions}
We have presented the first \emph{linear-time} {\FPT} algorithm for the {\labelcoverv} problem. The parameter-dependence in the running time of this algorithm is $2^{\Oh(k \cdot |\Sigma| \log |\Sigma|)}$. As a result, this algorithm improves upon that of Chitnis et al.~\cite{ChitnisCHPP12} (which has parameter-dependence $2^{\Oh(k^2 \cdot \log |\Sigma|)}$) with respect to the dependence on both the parameter as well as input-size when $|\Sigma|\leq k$. However, the best known parameter-dependence is $4^{k \log |\Sigma|}$ which  was obtained by  Wahslstr\"{o}m~\cite{Wahlstrom14}, albeit at a significantly higher dependence on the input-size. We leave open 
the question of finding the optimal dependence on the parameter while preserving linear dependence on the input-size. 

\bibliographystyle{siam}
\bibliography{literature}

\newpage
\appendix

\end{document}